\documentclass[sigconf]{acmart}

\usepackage{booktabs}
\usepackage{multirow}
\usepackage{listings}
\usepackage{balance}
\usepackage{graphicx}
\usepackage{tikz}
\usetikzlibrary{arrows.meta,calc,positioning,decorations.pathreplacing}
\usepackage{pifont}

\acmConference[SACMAT '26]{31st ACM Symposium on Access Control Models and Technologies}{July 8--10, 2026}{Waterloo, Canada}
\acmYear{2026}
\acmDOI{}
\acmISBN{}
\setcopyright{rightsretained}
\acmPrice{}

\begin{document}

\title{ProcRoute: Process-Scoped Authorization of Split-Tunnel Routes}

\author{Arul Thileeban Sagayam}
\affiliation{%
  \institution{Bloomberg}
  \city{New York}
  \country{USA}}
\email{asagayam@bloomberg.net}

\begin{abstract}
In most split-tunnel VPN/ZTNA deployments, installing an internal route
authorizes the entire device, not a specific application, to use it.
An unprivileged malicious process can therefore reach internal
services by reusing routes intended for corporate applications.

We present \textbf{ProcRoute}, a system that restricts
internal-route access to explicitly authorized applications.
ProcRoute models route access as an access-control problem:
application identities are principals,
destination prefixes with port and protocol constraints are
resources, and a total, default-deny decision function mediates
every \texttt{connect()} and UDP \texttt{sendmsg()} to an internal
destination.
Processes without a grant retain external access but are denied
internal routes under our threat model.

We describe ProcRoute's formal model, a Linux
prototype built on cgroup~v2 and eBPF socket-address hooks, and
two complementary evaluations.
In a two-machine WireGuard deployment,
ProcRoute matches the WireGuard baseline and
13\% faster than an nftables cgroup-matching configuration, with a
p50 connect latency of 93\,\textmu{}s (+3.6\,\textmu{}s over
baseline), flat policy scaling to 5\,000 prefixes, and
sub-millisecond revocation.
Single-machine loopback microbenchmarks confirm low hook overhead:
2.7\,\textmu{}s on the internal-allow path,
82/82 unauthorized pivot attempts blocked, and zero
transient allows across 1.2~million connection attempts during
policy reload.
\end{abstract}




\maketitle

\section{Introduction}
Remote work has made remote access VPNs and ZTNA agents default infrastructure. Many organizations enable split tunneling so endpoints can access local networks and the public Internet efficiently while still reaching internal resources. However, split tunneling weakens the boundary between ``trusted'' corporate routes and ``untrusted'' external networks: the endpoint becomes a bridge, and internal routes are trivially reachable by any local process~\cite{mitre_lateral}. Split tunneling persists in most enterprise deployments despite guidance discouraging it~\cite{nist_sp80046r2,nist_sp80053r5,nsa_cisa_vpn}, because full-tunnel routing imposes unacceptable latency, capacity costs, and user friction at scale; yet no standard mechanism on commodity desktop OSes restricts which processes may use VPN-installed routes.

This is well-understood. Guidance around VPN deployments explicitly calls out the risk of split-tunnel configurations that allow a host to communicate to the Internet while simultaneously connected to organizational resources, since the external connection can be used to attack or infiltrate the VPN-connected environment~\cite{nist_sp80046r2}.  Nevertheless, in practice enterprises use split tunneling for usability and throughput, particularly when only a subset of applications truly need internal access.

We argue that the core issue is \emph{granularity}. Enterprise
remote access makes a decision at the device level,
thus allowing routes to internal prefixes, and the operating system
then allows any application to use those routes.
Many endpoint threats  manifest as \emph{separate
unprivileged processes} (phishing-launched payloads, malicious
installers, adware bundles, and commodity malware running as the user)
that are distinct from any legitimate corporate
application yet inherit the same VPN routes.
A second class of threats, in-process code execution via browser
extensions, library injection, or memory corruption, operates within
an already-running application and is not distinguishable at process
granularity; we scope this class out of ProcRoute's threat model and
analyze it as a limitation later.
For the separate-process class, route authorization at process
granularity is a feasible enforcement boundary that current
split-tunnel deployments do not provide.

\textbf{ProcRoute} introduces process-scoped route authorization:
only processes bound to an authorized principal may reach internal
prefixes.
Other processes are denied access to internal routes but are not
restricted on external destinations, preserving split-tunnel Internet
access.
ProcRoute is designed to coexist with existing split-tunnel
configurations and to impose low overhead on authorized connections.
\subsection{Contributions}
Commodity OSes provide basic packet-filtering primitives, but they
don't treat VPN installed routes as resources, and no existing
tool provides formal policy semantics for this enforcement.
ProcRoute addresses this gap with three contributions:
(1)~A \emph{formal access-control model} that treats application
identities as principals, destination prefixes as resources, and
enforces default-deny through a total decision function with six
stated invariants, enabling security guarantees
(e.g., monotone composition, non-widening updates) that cannot be
stated for ad-hoc filter configurations
(Section~\ref{sec:model}).
(2)~\emph{Composable enforcement}: attaching the BPF program with
\texttt{BPF\_F\_ALLOW\_MULTI} yields monotone composition: descendant
cgroups can restrict but never expand the ancestor's allow set
(Section~\ref{sec:composition}).
(3)~A \emph{split client/gateway architecture} for WireGuard,
evaluated on a two-machine deployment and single-machine
loopback microbenchmarks
(Sections~\ref{sec:design}--\ref{sec:eval},
Appendix~\ref{sec:appendix:micro}).
\subsection{Delta vs.\ nftables and systemd}
Existing tools(nftables, systemd) can approximate per-process filtering,
but they are low-level packet-filter primitives.
ProcRoute adds formal policy semantics, monotone composition,
structured deny auditing, and optional per-binary hash verification
(detailed in Sections~\ref{sec:model}--\ref{sec:proc-identity}).

\section{Threat Model and Goals}
\label{sec:threat}
\subsection{Threat model}
Procroute is designed to reduce risk when an adversary obtains code execution as a \emph{separate, unprivileged process}~\cite{mitre_lateral}:
\begin{itemize}
  \item \textbf{Adversary:} a process running as the user that is \emph{distinct} from any authorized corporate application, for example a phishing-launched payload or adware executing as its own process.
  \item \textbf{Objective:} laterally move into internal networks by using VPN-installed routes.
  \item \textbf{Assumptions:} the OS kernel and ProcRoute enforcement are not compromised; the adversary does not have root/admin privileges sufficient to disable system security policy or load competing eBPF programs.
\end{itemize}

We explicitly \textbf{do not} claim to defend against:
\begin{itemize}
  \item a kernel-level adversary or a local admin who can remove the enforcement agent, or
  \item \emph{in-process code execution} within an already-authorized
    application (e.g., a malicious browser extension or exploitation of
    a memory-corruption vulnerability in an authorized binary).
    In-process compromise inherits the victim process's cgroup
    membership and therefore its network grants; ProcRoute cannot
    distinguish legitimate from injected code within a single address
    space. This limitation is inherent to any enforcement boundary
    drawn at process granularity.
\end{itemize}

\subsection{Goals}
ProcRoute is designed around four goals:
\begin{enumerate}
  \item \textbf{Process-level least privilege:} internal routes should be usable only by explicitly authorized applications.
  \item \textbf{Split-tunnel compatible:} preserve direct Internet access for the rest of the system and avoid forcing all traffic through a VPN gateway.
  \item \textbf{Low overhead:} enforcement should add minimal latency/CPU overhead for allowed flows.
  \item \textbf{Auditable decisions:} policy should be readable and decisions should be loggable for incident response and compliance.
\end{enumerate}

\section{ProcRoute Access-Control Model}
\label{sec:model}

We now formalize ProcRoute as an access-control system, defining
principals, resources, the policy language, the enforcement decision
function, and the security invariants the system must maintain.

ProcRoute supports two deployment modes that share the same
principal-binding and grant-evaluation machinery but differ in
where the deny verdict is applied:
\begin{itemize}
  \item \textbf{Local-only enforcement.}
    Cgroup socket-address hooks evaluate the decision function
    directly on the endpoint, returning \texttt{EPERM} for
    unauthorized connections
    (Sections~\ref{sec:decision}--\ref{sec:invariants}).
  \item \textbf{Split client/gateway enforcement.}
    The same cgroup hooks resolve the originating principal on
    the client but \emph{tag} the socket instead of denying;
    a cooperating gateway enforces per-peer policy over a
    WireGuard tunnel
    (Section~\ref{sec:model:split}).
\end{itemize}
Both modes may operate simultaneously: local enforcement
provides immediate \texttt{EPERM} feedback, while gateway
enforcement adds a second, network-level enforcement point.

\subsection{Principals}

A \emph{principal} in ProcRoute is an application identity
$p \in \mathcal{P}$, where $\mathcal{P}$ is a finite set of
identifiers defined by administrative policy.
Each principal corresponds to a logical application that may
comprise one or more OS processes (e.g., a browser and its helper
processes).

Principal binding is realized through cgroup membership.
Let $\mathcal{C}$ denote the set of all cgroup nodes in the
hierarchy.  The ProcRoute controller creates a dedicated cgroup
subtree for each $p \in \mathcal{P}$ and moves the application's
processes into it.  We define two total functions:
\begin{itemize}
  \item $\mathit{cg}: \mathit{PID} \to \mathcal{C}$ maps every
    live process to its unique cgroup,
  \item $\mathit{bind}: \mathcal{C} \to \mathcal{P} \cup \{\bot\}$
    maps each cgroup to its bound principal, or to $\bot$ if the
    cgroup has no binding.  Only finitely many cgroups map to a
    value in $\mathcal{P}$; all others map to $\bot$.
\end{itemize}
The effective principal of a process is the composition
$\mathit{prin}(\mathit{pid}) = \mathit{bind}(\mathit{cg}(\mathit{pid}))$.
Because both $\mathit{cg}$ and $\mathit{bind}$ are total,
$\mathit{prin}$ is total: every process has a well-defined
principal (or $\bot$).

A process whose cgroup has no binding
($\mathit{prin}(\mathit{pid}) = \bot$) is \emph{unauthorized} and
receives no grants.  This is the common case: the vast majority
of cgroups on a system are unbound, and the vast majority of
processes are therefore unauthorized with respect to ProcRoute.

\paragraph{Policy updates.}
$\mathcal{P}$, $\mathcal{I}$, and $\mathcal{R}$ are parameters of
a policy \emph{instance}.  The controller may atomically replace
the contents of the BPF maps to install a new policy instance
(e.g., when the VPN session reconnects with different route
grants).  The model treats each instance as a static
configuration; transitions between instances are handled by the
controller and fall outside the formal model here.

\paragraph{Revocation semantics.}
Revoking access for principal~$p$ at time~$t$ involves two
guarantees of decreasing strength:
\begin{itemize}
  \item \textbf{G1 (connect-time):}
    Every \texttt{connect()} or first \texttt{sendmsg()}
    call issued by~$p$ after time~$t$ to a destination
    outside $R_p^{\,\prime}$ (the updated grant set) is
    denied.
  \item \textbf{G2 (data-path):}
    Every \emph{existing} socket opened by~$p$ to a
    destination in $R_p \setminus R_p^{\,\prime}$ is
    terminated within a bounded interval~$\delta$
    after time~$t$.
\end{itemize}
G1 is immediate whenever the BPF maps are updated, because each
hook re-evaluates the policy on every call
(Section~\ref{sec:decision}).
G2 requires an \emph{authorization epoch}: the controller
increments a global epoch counter $\varepsilon$ in a BPF array
on each policy update.
Sockets authorized at epoch~$\varepsilon' < \varepsilon$ are
\emph{stale}.  The prototype implements two revocation paths
for stale sockets:
(a)~the \texttt{sendmsg} hook checks the epoch recorded per
socket cookie and denies further sends (immediate for UDP);
(b)~a background sweeper terminates stale TCP sockets via
\texttt{ss~--kill} (\texttt{SOCK\_DESTROY}) at a configurable
interval.

\subsection{Resources}

A \emph{resource} is a network destination reachable via a
split-tunnel route.  Each resource is a tuple
\[
  r = \langle \mathit{prefix},\; \mathit{proto},\;
        [\mathit{port\_lo},\, \mathit{port\_hi}] \rangle
\]
where $\mathit{prefix}$ is an IPv4 or IPv6 CIDR prefix,
$\mathit{proto} \in \{\mathrm{TCP},\, \mathrm{UDP},\, \mathord{*}\}$
specifies the transport protocol ($\mathord{*}$ matches both TCP
and UDP), and
$[\mathit{port\_lo},\, \mathit{port\_hi}]$ is an inclusive port range.
The sentinel value $[0,0]$ denotes ``all ports'': when
$\mathit{port\_lo} = 0$, the port constraint is vacuously
satisfied regardless of the destination port.
The protocol and port wildcards compose independently:
$\langle \mathit{prefix},\, \mathord{*},\, [0,0] \rangle$ covers
every TCP and UDP connection to any port within
$\mathit{prefix}$.

Resources are \emph{flat}: there is no hierarchy or subsumption
among them; the policy is simply a set of grants.  Overlapping prefix ranges are not given any special
precedence in the policy language (coverage is defined by the covering
relation below).  For efficient kernel enforcement via a single LPM
lookup, ProcRoute additionally assumes a well-formed,
destination-disjoint normalized grant set.

\paragraph{Overlaps and LPM resolution.}
\label{sec:lpm-overlaps}
The formal decision function permits multiple grants
in $R_p$ to cover the same destination tuple, yielding a union-of-grants semantics.
The kernel prototype uses a per-principal longest-prefix-match (LPM) trie for
efficiency (Section~\ref{sec:enforcement-pipeline}) and therefore obtains at most
one prefix match for a given destination IP.
To reconcile these, ProcRoute assumes (and the controller can validate) a
\emph{well-formed} grant set: after normalizing duplicates, each principal's
resources are \emph{destination-disjoint}: for any destination IP~$d$, at most
one $r\in R_p$ has $d\in r.\mathit{prefix}$.
Overlapping entries with identical protocol/port constraints are merged; overlaps
with incompatible constraints are rejected as ambiguous.
This rules out policies that assign different port constraints to
overlapping prefixes for the same principal (e.g., \texttt{10.0.0.0/8}
on port~443 and \texttt{10.1.0.0/16} on port~22); enterprise firewall
audits suggest 15--25\,\% of zone-based rulesets contain such
entries~\cite{wool2004quantitative,wool2010trends}.
A \emph{cascaded trie} extension that lifts this restriction with one
additional map lookup is described in Appendix~\ref{sec:appendix:q6}.
Under the destination-disjoint condition, the existential check $\exists r\in R_p:\mathit{dst}\sqsubseteq r$
is equivalent to a single LPM lookup followed by the matched rule's
protocol/port predicate.

We write $\mathit{dst} \sqsubseteq r$ to indicate that a concrete
connection attempt with destination IP~$d$, protocol~$t$, and
port~$q$ is \emph{covered} by resource~$r$:
\begin{multline*}
  \mathit{dst} \sqsubseteq r \;\iff\;
    d \in r.\mathit{prefix}
    \;\land\;
    (r.\mathit{proto} = \mathord{*} \lor t = r.\mathit{proto})\\
    \;\land\;
    (r.\mathit{port\_lo} = 0
      \lor r.\mathit{port\_lo} \le q \le r.\mathit{port\_hi})
\end{multline*}

\subsection{Internal route set}

The administrator defines a finite set of internal prefixes
$\mathcal{I} = \{i_1, i_2, \ldots, i_n\}$, each an IP CIDR block.
A destination IP~$d$ is \emph{internal} iff
$d \in \bigcup_{i \in \mathcal{I}} i$.
Destinations outside $\mathcal{I}$ are considered external and are
not subject to ProcRoute enforcement; they pass through
unconditionally to preserve split-tunnel Internet access.

\subsection{Policy}
\label{sec:policy}

A ProcRoute policy $\Pi$ is a triple
$\Pi = \langle \mathcal{I},\; \mathcal{P},\; \mathcal{R} \rangle$,
where:
\begin{itemize}
  \item $\mathcal{I}$ is the internal route set (set of CIDR prefixes),
  \item $\mathcal{P}$ is the set of application principals, and
  \item $\mathcal{R}: \mathcal{P} \to 2^{\mathit{Resource}}$ assigns
    each principal a (possibly empty) set of allowed resources.
\end{itemize}

We write $R_p = \mathcal{R}(p)$ for the allow set of principal~$p$.

ProcRoute's policy language is \emph{purely positive}: grants
authorize access; there are no deny rules, exceptions, or negative
grants.  A principal's effective access is exactly
$\bigcup_{r \in R_p} \{ \mathit{dst} \mid \mathit{dst} \sqsubseteq r \}$, the
union of destinations covered by its grants.  There is no mechanism
to deny a sub-range within a broader grant; restricting a
principal's access requires narrowing its grant set.

The policy is \emph{default-deny}: this follows directly from the
structure of the decision function.
A connection to an internal destination is permitted only if the
process is bound to a principal and an explicit matching grant
exists.  In the absence of any grant, the decision function
returns $\mathit{false}$, not because a deny rule fires, but
because no allow rule matches.
External destinations require no grant and are always permitted.

\subsection{Decision function}
\label{sec:decision}

Let $\mathit{dst} = \langle d, t, q \rangle$ denote a connection
attempt to destination IP~$d$, transport protocol~$t$, and
port~$q$, issued by a process with identifier $\mathit{pid}$.
ProcRoute evaluates the following decision function at each
\texttt{connect()} or UDP \texttt{sendmsg()} invocation:

\[
  \mathit{allow}(\mathit{pid},\, \mathit{dst}) =
  \begin{cases}
    \mathit{true}
      & \text{if } d \notin \bigcup_{i \in \mathcal{I}} i
        \;\text{(external)}\\[4pt]
    \mathit{true}
      & \text{if } d \in \bigcup_{i \in \mathcal{I}} i
        \;\land\; \mathit{prin}(\mathit{pid}) \neq \bot \\
      & \quad \land\; \exists\, r \in R_{\mathit{prin}(\mathit{pid})}
              :\; \mathit{dst} \sqsubseteq r \\[4pt]
    \mathit{false}
      & \text{otherwise}
  \end{cases}
\]

The function is \emph{total}: every $(\mathit{pid},\, \mathit{dst})$
pair yields exactly one verdict.
Totality follows from the exhaustive case split on membership
in~$\mathcal{I}$, the ``otherwise'' clause, and the totality of
$\mathit{prin}$ (Section~\ref{sec:model}), and is a prerequisite
for complete mediation (Section~\ref{sec:invariants}).

When $\mathit{allow}$ returns $\mathit{false}$, the kernel returns
\texttt{EPERM} to the calling process and emits a deny event to the
audit log.  When it returns $\mathit{true}$, the socket operation
proceeds unmodified.

The function is evaluated in the kernel by the eBPF programs
attached to the cgroup socket-address hooks, using a sequence of
BPF map lookups (Section~\ref{sec:enforcement-pipeline}).

\subsection{Split architecture: client tagging and gateway enforcement}
\label{sec:model:split}

The local decision function $\mathit{allow}$
(Section~\ref{sec:decision}) mediates every socket operation on a
single host.
When ProcRoute is deployed across a WireGuard tunnel, the same
cgroup hooks that implement $\mathit{allow}$ are reused on the
client to resolve $\mathit{prin}(\mathit{pid})$, but the verdict
changes from deny to \emph{tag}: the client stamps the socket
with the originating principal, and a cooperating \emph{gateway}
enforces a per-peer policy based on the tag and the packet's
tunnel source address.
This subsection extends the formal model to cover both sides.

\paragraph{Per-peer policy.}
A gateway serves $K$ WireGuard peers, each identified by its
tunnel IP address $s_k$.
The gateway maintains a policy
$\Pi_{\mathit{gw}} =
  \langle \mathcal{I},\;
         \{s_1,\ldots,s_K\},\;
         \mathcal{R}_{\mathit{gw}} \rangle$,
where $\mathcal{R}_{\mathit{gw}}(s_k, a)$ is the grant set
for app\_index~$a$ on peer~$s_k$.
Different peers may have different entitlements: a managed
corporate laptop may receive grants for sensitive prefixes that
a contractor device does not.
The controller distributes $\mathcal{R}_{\mathit{gw}}$ to the
gateway, which populates its \texttt{app\_allow} LPM trie with
composite keys that include the peer's tunnel IP
(Section~\ref{sec:design:split}).

\paragraph{Client-side tagging.}
The client hooks (\texttt{cgroup/connect4/6}) resolve
$\mathit{prin}(\mathit{pid})$ as before but \emph{never
deny}: instead of returning \textsc{deny} for unbound
processes, the hook stamps the socket with
\texttt{SO\_MARK} encoding the principal identity.
Let $\mathit{tag}: \mathit{PID} \to \mathbb{N} \cup \{0\}$ be
the tagging function:
\[
  \mathit{tag}(\mathit{pid}) =
  \begin{cases}
    \mathit{app\_index}(p) & \text{if } \mathit{prin}(\mathit{pid}) = p \neq \bot \\
    0 & \text{if } \mathit{prin}(\mathit{pid}) = \bot
  \end{cases}
\]
A TC egress BPF program on \texttt{wg0} encodes the tag into
packet headers: for IPv4, the IP identification field
(\texttt{ip->id}) carries $\mathit{app\_index}$ and the TOS byte
carries the authorization epoch~$\varepsilon$; for IPv6, the flow
label carries $\mathit{app\_index}$ and the traffic class carries
$\varepsilon$.
The client hooks are purely annotative; they impose no
connectivity restrictions.

\paragraph{Gateway-side enforcement.}
The gateway executes a TC ingress BPF program on its \texttt{wg0}
interface that evaluates a total function
$\mathit{allow\_gw}$ on every inbound packet before routing.
Let $s$ denote the inner source IP (identifying the WireGuard
peer) and $a = \mathit{tag}(\mathit{pkt})$ the app\_index
extracted from the packet header:
\[
  \mathit{allow\_gw}(\mathit{pkt},\, \mathit{dst}) =
  \begin{cases}
    \mathit{false}
      & \text{if } a = 0
        \;\text{(untagged)} \\[4pt]
    \mathit{true}
      & \text{if } d \notin \bigcup_{i \in \mathcal{I}} i
        \;\text{(external)} \\[4pt]
    \mathit{true}
      & \text{if } d \in \bigcup_{i \in \mathcal{I}} i \\
      & \quad \land\; \exists\, r \in
              \mathcal{R}_{\mathit{gw}}(s,\, a)
              :\; \mathit{dst} \sqsubseteq r \\[4pt]
    \mathit{false}
      & \text{otherwise}
  \end{cases}
\]
The function is total: every packet receives a determinate
verdict.  Denied packets are dropped via \texttt{TC\_ACT\_SHOT}
before reaching the routing layer.
Unlike local enforcement, which returns \texttt{EPERM} at
\texttt{connect()} time, gateway-only denies are silent drops;
clients observe standard transport timeouts (TCP SYN
retransmission, UDP undeliverability) unless local enforcement
is also enabled (Section~\ref{sec:design:split}).
Because the grant lookup is keyed on
$(s,\, a)$, two peers may carry the same app\_index~$a$ yet
receive different grants.

\paragraph{Tag integrity and trust model.}
The split architecture relies on three properties:
\begin{enumerate}
  \item \emph{Tag authenticity.}
    WireGuard provides authenticated encryption of the entire inner IP packet.
    An off-path adversary cannot
    inject or modify tags in transit.
  \item \emph{Tag binding.}
    Setting \texttt{SO\_MARK} on a socket requires
    \texttt{CAP\_NET\_\allowbreak ADMIN}, which the threat model
    excludes for the adversary
    (Section~\ref{sec:threat}).  The BPF hooks set the mark in
    kernel context on behalf of the process; an unprivileged
    process cannot override it.
  \item \emph{Gateway mediation.}
    The TC ingress program on \texttt{wg0} executes on every
    packet delivered by the WireGuard interface, before the kernel
    routes the packet.  No bypass path exists for traffic arriving
    on \texttt{wg0}.
\end{enumerate}

\paragraph{End-to-end security property.}
The composition of client tagging and gateway enforcement preserves
the reachable-set guarantee of the local-only model:
an unauthorized process ($\mathit{prin}(\mathit{pid}) = \bot$)
cannot reach any internal destination via the tunnel, because
(i)~$\mathit{tag}(\mathit{pid}) = 0$ (untagged), and
(ii)~the gateway drops all untagged packets.
For a bound process with principal~$p$ on peer~$s_k$, the gateway
allows only destinations covered by
$\mathcal{R}_{\mathit{gw}}(s_k,\, \mathit{app\_index}(p))$.
When the gateway grant set is a subset of the client's local
grant set for every principal (the expected configuration, since
the gateway may impose additional restrictions), the reachable set
under the split architecture is bounded by the local model:
\[
  \mathit{Reach}_{\mathit{split}}(s_k, p)
  \;\subseteq\;
  \mathit{Reach}_{\mathit{local}}(p)
  \;=\;
  \bigcup_{r \in R_p} \{ \mathit{dst} \mid \mathit{dst} \sqsubseteq r \}
\]

\subsection{Policy composition and monotonicity}
\label{sec:composition}

ProcRoute attaches its BPF enforcement program using the
\texttt{BPF\_F\_\allowbreak ALLOW\_MULTI} flag, which composes
programs across the cgroup hierarchy \emph{conjunctively}: the
kernel evaluates every attached program from ancestor to descendant,
and the connection is allowed only if \emph{all} programs return
allow.

Let $\mathit{allow}_A$ denote the decision function of a BPF
program attached at ancestor cgroup~$A$, and let
$\mathit{allow}_C$ denote the decision function of a supplementary
program attached at a descendant cgroup $C \subseteq A$.
The effective verdict for a process in~$C$ is:
\[
  \mathit{allow}_{\mathit{eff}}(\mathit{pid},\, \mathit{dst})
  = \mathit{allow}_A(\mathit{pid},\, \mathit{dst})
    \;\land\;
    \mathit{allow}_C(\mathit{pid},\, \mathit{dst})
\]

\paragraph{Assumption (filter-only programs).}
cgroup socket-address programs can rewrite the
destination address/port in the \texttt{bpf\_sock\_addr} context.
ProcRoute's enforcement program is a pure filter and does not modify
the destination tuple.
The monotonicity claim below is stated for the common deployment in
which co-attached descendant programs are also filter-only;
under this assumption, reachability
is determined solely by the conjunction of allow/deny verdicts.

\begin{theorem}[Monotonicity of \texttt{BPF\_F\_ALLOW\_MULTI}]
\label{thm:monotonicity}
Under conjunctive composition, the set of destinations reachable
by a process in descendant cgroup~$C \subseteq A$ is a subset of
those reachable under the ancestor's policy alone:
\begin{multline*}
  \{ \mathit{dst} \mid
     \mathit{allow}_{\mathit{eff}}(\mathit{pid},\, \mathit{dst})
     = \mathit{true} \}\\
  \;\subseteq\;
  \{ \mathit{dst} \mid
     \mathit{allow}_A(\mathit{pid},\, \mathit{dst})
     = \mathit{true} \}
\end{multline*}
A program attached at a child cgroup can further restrict but
never expand the ancestor's allow set.
\end{theorem}

\begin{proof}[Proof sketch]
By definition, $\mathit{allow}_{\mathit{eff}} =
\mathit{allow}_A \land \mathit{allow}_C$.
For any $(\mathit{pid}, \mathit{dst})$ pair where
$\mathit{allow}_{\mathit{eff}} = \mathit{true}$, both
$\mathit{allow}_A = \mathit{true}$ and
$\mathit{allow}_C = \mathit{true}$ must hold (conjunction).
Therefore every element of the effective allow set is in
$A$'s allow set, establishing the subset relation.
The argument extends by induction to any finite chain of
cgroups $C_1 \subseteq C_2 \subseteq \cdots \subseteq A$:
each additional conjunct can only shrink the allow set.
\end{proof}

This property is critical for ProcRoute's security posture: the
default-deny baseline established by the ancestor program cannot
be weakened by any attachment in a descendant cgroup, even if
the descendant's program is controlled by a different
administrative entity (e.g., an application-specific audit hook
added by a team-level policy).
The implementation-level implications of this composition model,
including the contrast with \texttt{BPF\_F\_ALLOW\_OVERRIDE}
(which would break monotonicity), are discussed in
Section~\ref{sec:attach-flags}.
Note that if a descendant program rewrites the destination tuple,
it can change what the kernel ultimately connects to; such rewriting
is orthogonal to (and can undermine) verdict-only monotonicity.
ProcRoute is designed and deployed as a non-mutating filter.

\begin{theorem}[Non-widening policy updates]
\label{thm:update-safety}
Let $\Pi = \langle \mathcal{I}, \mathcal{P}, \mathcal{R} \rangle$
denote the current policy and
$\Pi' = \langle \mathcal{I}', \mathcal{P}', \mathcal{R}' \rangle$
an updated policy installed by the controller via a sequence of
\texttt{bpf\_map\_\allowbreak update\_elem} and
\texttt{bpf\_map\_\allowbreak delete\_elem} calls.
If the controller applies updates using a
\emph{fail-closed} ordering: (i)~for maps whose absence
implies \textsc{deny} on the internal path (e.g.,
\texttt{cgroup\_to\_app}, \texttt{app\_allow},
\texttt{app\_exec\_hash}, \texttt{task\_verified}), it
\emph{deletes} removed entries before \emph{inserting} new
entries; and (ii)~for the \texttt{internal\_\allowbreak
prefixes} classifier (where a miss implies external
passthrough \textsc{allow}), it \emph{inserts} new prefixes
before \emph{deleting} removed prefixes, then at every
intermediate map state $M_k$ observed by a concurrent BPF
hook invocation, the effective allow set is a subset of
$\mathit{Allow}(\Pi) \cup \mathit{Allow}(\Pi')$:
\[
  \mathit{Allow}(M_k)
  \;\subseteq\;
  \mathit{Allow}(\Pi) \cup \mathit{Allow}(\Pi')
\]
In particular, no destination that is denied by \emph{both} the
old and new policies is transiently permitted during the update.
\end{theorem}

\begin{proof}[Proof sketch]
Each map operation is atomic.
For \emph{deny-biased} maps (\texttt{cgroup\_to\_app},
\texttt{app\_allow}, etc.), a missing entry returns \textsc{deny}; deleting before inserting
can only shrink the allow set before expanding it, so no transient
allow absent from both $\Pi$ and $\Pi'$ can appear.
For the \emph{allow-biased} \texttt{internal\_prefixes} map (miss
$\Rightarrow$ external passthrough), inserting before deleting keeps
the classified-internal set as a superset of both $\mathcal{I}$ and
$\mathcal{I}'$, preventing transient external reclassification.
Combining both orderings with the total decision function yields
$\mathit{Allow}(M_k) \subseteq \mathit{Allow}(\Pi) \cup
\mathit{Allow}(\Pi')$.
\end{proof}

\subsection{Enforcement pipeline}
\label{sec:enforcement-pipeline}

The decision function is implemented as a six-stage pipeline
executed inside the BPF hook on each socket address operation:

\begin{enumerate}
  \item \textbf{Destination extraction.}
    Read the destination IP, port, and protocol from the
    \texttt{bpf\_sock\_addr} context.
  \item \textbf{Internal prefix lookup.}
    Query the \texttt{internal\_prefixes} LPM trie with the
    destination IP.  If no match, return \textsc{allow} (external
    destination; split-tunnel passthrough).
  \item \textbf{Principal resolution.}
    Read the cgroup~ID of the calling process and look
    up \texttt{cgroup\_to\_app}.  If no binding exists
    ($\mathit{prin} = \bot$), return \textsc{deny}.
  \item \textbf{Binary hash verification gate} (optional).
    If the principal's policy includes an \texttt{exec\_hash}
    requirement, look up the calling process's TGID in the
    \texttt{task\_verified} map.  If no entry exists, or the
    verified principal does not match, return \textsc{deny}.
    This step is skipped for principals without a declared hash,
    imposing zero overhead on the common case.
  \item \textbf{Per-app prefix lookup.}
    Construct a composite key $\langle p,\, d \rangle$ and query the
    \texttt{app\_allow} LPM trie.  If no matching rule exists, return
    \textsc{deny}.
  \item \textbf{Port and protocol check.}
    Verify that the destination port and protocol satisfy the matched
    rule's constraints.  If not, return \textsc{deny}; otherwise
    return \textsc{allow}.
\end{enumerate}

Every \textsc{deny} emits a structured event (PID, command
name, cgroup~ID, principal index, destination IP, port, protocol,
timestamp) to a ring buffer for user-space consumption.

\subsection{Security invariants}
\label{sec:invariants}

ProcRoute maintains six invariants, grouped into local enforcement and
gateway enforcement:

\paragraph{Local enforcement.}
(1)~\textbf{Default-deny}: an internal connection requires both a bound principal
and an explicit matching grant;
(2)~\textbf{Unimpeded external access}: destinations outside
$\mathcal{I}$ are never blocked;
(3)~\textbf{Complete mediation}: every \texttt{connect()}/\texttt{sendmsg()} in the ProcRoute subtree is evaluated before
the kernel processes the address~\cite{kernel_bpf_cgroup,kernel_cgroup2};
(4)~\textbf{Audit completeness}: every deny emits a structured event.

\paragraph{Gateway enforcement.}
(5)~\textbf{Tag integrity}: unprivileged processes cannot forge
\texttt{SO\_MARK} values, because setting the mark requires
\texttt{CAP\_NET\_ADMIN} (excluded by the threat model);
the BPF hooks set the mark in kernel context and an
unprivileged adversary cannot override it.
(6)~\textbf{Gateway mediation}: the TC ingress BPF program on
\texttt{wg0} evaluates every packet arriving on the tunnel
interface before the kernel routes the packet; untagged or
policy-violating packets are dropped via \texttt{TC\_ACT\_SHOT}.
Failure modes are analyzed in Sections~\ref{sec:residual}
and~\ref{sec:discussion}.

\subsection{Process identity and executable hash verification}
\label{sec:proc-identity}

Process identity in ProcRoute is determined by cgroup membership:
the controller creates a cgroup subtree per principal and moves
application processes into it at launch time.  The BPF program
calls \texttt{bpf\_get\_current\_cgroup\_id()} to obtain a stable,
kernel-assigned identifier and looks it up in
\texttt{cgroup\_to\_app}.

Because child processes inherit cgroup membership, a child that
\texttt{exec()}s a different binary retains the parent's grants
by default.  To tighten this, the prototype implements optional
per-binary SHA-256 hash verification.
When a principal's policy includes an \texttt{exec\_hash} field,
a BPF program on the \texttt{sched\_process\_exec} tracepoint
invalidates the prior verification and emits an exec event.
The user-space daemon reads the event, computes the binary's
SHA-256 digest via \texttt{/proc/<tgid>/exe}, and on match sets
\texttt{task\_verified[tgid] = app\_index}; on mismatch, the
process remains denied.
A companion tracepoint on \texttt{sched\_process\_exit} deletes
the TGID entry, preventing stale verifications.
Connections attempted before verification completes are denied.
Principals without a declared hash skip the gate entirely,
imposing zero overhead on the common case.

\section{Design and Implementation}
\label{sec:design}

ProcRoute comprises three components: (1)~a \emph{policy authority}
on the enterprise control plane that defines internal route sets and
per-application grants; (2)~an \emph{endpoint controller} (user-space
daemon) that manages cgroups, loads BPF programs, populates policy
maps, and streams audit logs; and (3)~\emph{kernel enforcement}
via eBPF programs attached to cgroup socket-address hooks.
Figure~\ref{fig:deployment} shows the overall deployment
architecture.

\subsection{Enforcement hooks and policy maps}
\label{sec:lifecycle}

ProcRoute attaches eBPF programs to four cgroup~v2
socket-address hooks: \texttt{connect4/6} and
\texttt{sendmsg4/6}.
These hooks execute synchronously inside the
\texttt{connect()} (or UDP \texttt{sendmsg()}) system-call
path, before the kernel constructs or transmits any packet.
On deny, the kernel returns \texttt{EPERM}; no SYN
segment or UDP datagram is emitted.

The BPF program interacts with six map structures to implement the
enforcement pipeline:

\begin{itemize}
  \item \texttt{internal\_prefixes}: an LPM trie encoding the
    internal prefix set~$\mathcal{I}$.  A miss means the destination
    is external and the connection is allowed without further checks.
  \item \texttt{cgroup\_to\_app}: a hash map that resolves the
    calling process's cgroup~ID to an application principal.
    A miss means no principal binding exists; the connection is
    denied.
  \item \texttt{app\_exec\_hash}: a hash map keyed by principal
    index whose values are 32-byte SHA-256 digests.  Presence of
    an entry means the principal requires binary hash
    verification; no key is present if verification is not
    required.
  \item \texttt{task\_verified}: a hash map keyed by TGID whose value is the verified
    principal index.  Set by the user-space daemon after
    confirming the binary's SHA-256 hash matches the policy.
    The BPF hook consults this map only for principals with an
    entry in \texttt{app\_exec\_hash}.
  \item \texttt{app\_allow}: a per-principal LPM
    trie keyed by (principal, destination) encoding
    the normalized grant set.  Under the
    well-formedness condition of
    Section~\ref{sec:lpm-overlaps}, at most one rule
    matches a given destination~IP, so a single LPM
    lookup plus a port/protocol check implements the
    grant predicate; otherwise deny.
\end{itemize}

\noindent
Every deny emits a structured audit event to a ring buffer.  A second ring buffer
(\texttt{exec\_events}) carries
\texttt{sched\_\allowbreak process\_exec} notifications to
user space for hash verification
(Section~\ref{sec:proc-identity}).

\subsection{Attach strategy and composition}
\label{sec:attach-flags}

ProcRoute attaches its enforcement program to the session-level
cgroup using \texttt{BPF\_F\_ALLOW\_MULTI}.  Under this flag,
programs \emph{accumulate} across the cgroup hierarchy: the kernel
runs every attached program from ancestor to descendant, and the
connection is allowed only if all return allow.  A descendant cgroup
therefore cannot weaken an ancestor's deny; it can only impose
additional restrictions.  This provides two properties:
(1)~\emph{fail-closed}: a process with no principal binding is
denied at the parent level before any child program executes; and
(2)~\emph{composability}: supplementary programs (e.g., per-app
audit hooks) can be layered in child cgroups without weakening the
baseline deny.

Policy differentiation across applications is achieved not by
attaching different BPF programs per child cgroup, but by
populating \texttt{cgroup\_to\_app}: each application cgroup's ID
maps to a principal whose grant set determines allowed destinations.
Processes in cgroups with no map entry receive no grants.

\subsection{Controller responsibilities}

The user-space controller daemon performs five functions:
(1)~provisions a cgroup subtree for each application principal
defined in the policy;
(2)~places application processes into their designated cgroups at
launch time (or moves them on policy update);
(3)~loads BPF programs, attaches them, and populates the
policy maps including per-principal expected hashes;
(4)~runs a hash verifier that reads exec events from the
\texttt{exec\_events} ring buffer, computes the SHA-256 digest
of the binary,
and on match sets the \texttt{task\_verified} map entry to
ungate the process's connections; and
(5)~drains the deny-events ring buffer and streams structured
audit records to a local sink.
Map updates take effect atomically from the BPF program's perspective, since each
map operation is a single \texttt{bpf\_map\_update\_elem} call.
In addition, the controller (i) normalizes each principal's grant set
to satisfy the destination-disjointness condition of
Section~\ref{sec:lpm-overlaps} before populating \texttt{app\_allow},
and (ii) applies the fail-closed update ordering of
Theorem~\ref{thm:update-safety} (insert-before-delete for
\texttt{internal\_prefixes}, delete-before-insert for deny-biased
maps).

\begin{figure*}[t]
  \centering
  \includegraphics[width=0.7\textwidth]{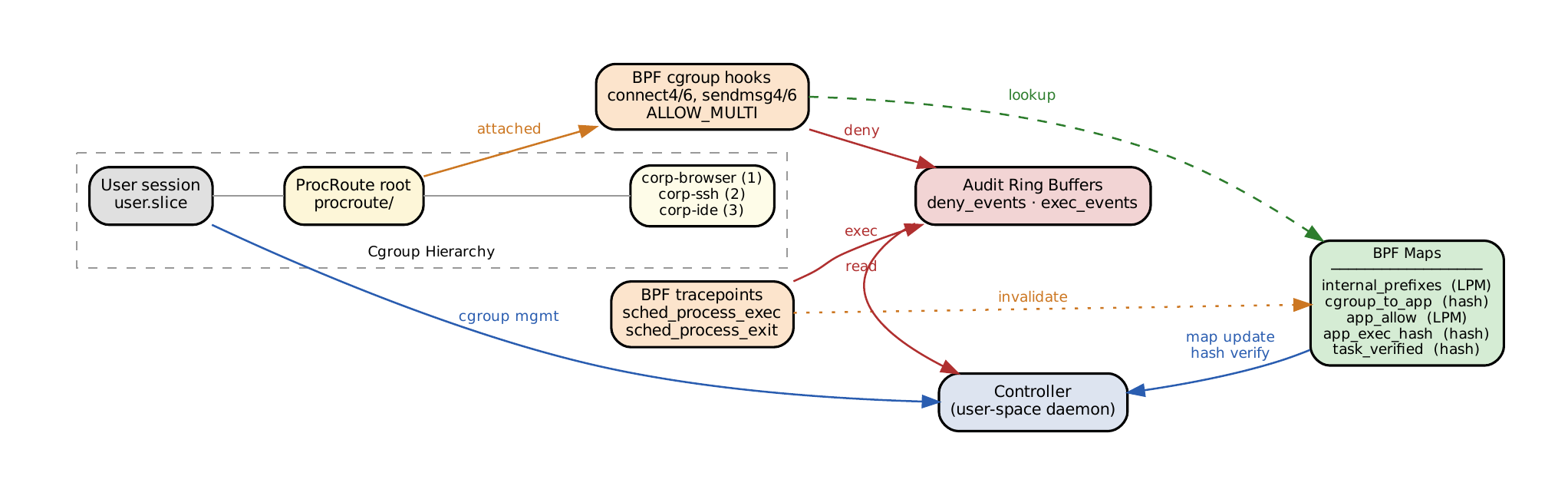}
  \caption{ProcRoute enforcement architecture: controller, BPF hooks, maps, and audit path.}
  \Description{Block diagram showing the ProcRoute deployment architecture with an endpoint controller managing cgroup placement, BPF program loading, and policy map population, eBPF hooks on cgroup socket-address operations enforcing per-process route authorization, and audit log streaming.}
  \label{fig:deployment}
\end{figure*}

\subsection{Split architecture: client tagging and gateway enforcement}
\label{sec:design:split}

When ProcRoute protects traffic traversing a WireGuard tunnel,
the enforcement boundary splits between the client endpoint and
the gateway.  The client reuses the same cgroup
\texttt{connect4/6} hooks described in
Section~\ref{sec:lifecycle}, but instead of returning
\textsc{deny} for unbound processes, the hooks annotate
outbound packets with the originating principal; the gateway
then enforces policy on the annotated packets before routing
them into the internal network.
This subsection describes the four components of the split
architecture.

\subsubsection{Client-side tagging}

The cgroup \texttt{connect4}/\texttt{connect6} hooks on the client
resolve $\mathit{prin}(\mathit{pid})$ using the same
\texttt{cgroup\_to\_app} map as the local-only mode but instead of
denying unbound processes, the hooks stamp the socket with
\texttt{SO\_MARK} encoding the principal's \texttt{app\_index}.
A TC egress BPF program on the \texttt{wg0} interface copies the
mark into packet header fields: for IPv4, \texttt{ip->id} carries
\texttt{app\_index} and \texttt{ip->tos} carries the authorization
epoch~$\varepsilon$; for IPv6, the flow label carries
\texttt{app\_index} and the traffic class carries~$\varepsilon$.
Unbound processes ($\mathit{prin} = \bot$) receive
$\texttt{SO\_MARK} = 0$; their packets carry tag zero and are
dropped at the gateway.
The client hooks never deny: connectivity restrictions are deferred
entirely to the gateway, avoiding duplicate enforcement and
simplifying client-side auditing.

\subsubsection{Gateway-side enforcement}

The gateway runs a TC ingress BPF program on its \texttt{wg0}
interface.  The program executes the following steps on every
inbound packet, before the kernel routes the packet:
\begin{enumerate}
  \item \textbf{Tag extraction.}
    Read \texttt{app\_index} and epoch from the IPv4 identification
    / TOS fields (or IPv6 flow label / traffic class).
    Record the inner source IP~$s$ (the peer's tunnel address).
  \item \textbf{Untagged drop.}
    If \texttt{app\_index} $= 0$, return \texttt{TC\_ACT\_SHOT}.
  \item \textbf{Per-peer, per-principal LPM lookup.}
    Construct the composite key
    $\langle s,\, \texttt{app\_index},\, \mathit{dst\_ip} \rangle$
    and query \texttt{app\_allow}.  If no match, drop.
    The three-part key ensures that two peers carrying the same
    \texttt{app\_index} are evaluated against their respective
    grant sets.
  \item \textbf{Port and protocol check.}
    Verify the destination port and protocol against the matched
    rule's constraints.  On mismatch, drop.
  \item \textbf{Tag clearing.}
    On allow, clear the tag fields (reset \texttt{ip->id},
    \texttt{tos}) so that internal hosts see clean headers.
    Emit a deny event to the ring buffer on drop.
\end{enumerate}
Exempt ports (e.g., WireGuard keepalive, DNS to the gateway's
own resolver) are whitelisted before step~2 to avoid breaking
control-plane traffic.

\subsubsection{Flow cache}

To amortize the cost of per-packet LPM lookups on long-lived
flows, the gateway maintains an LRU hash map keyed on the 5-tuple
(source IP, destination IP, source port, destination port,
protocol).
The first packet of each flow executes the full lookup pipeline
and, on allow, inserts an entry with a 5-second TTL.
Subsequent packets matching the cached 5-tuple skip directly to
the allow verdict.
The cache is invalidated atomically on epoch bump: the program
compares the cached epoch against the current global
epoch~$\varepsilon$ and treats stale entries as misses,
forcing a fresh lookup.
Section~\ref{sec:eval:wireguard} reports cache hit rates of 75\%
for short flows and $>$99.99\% for long flows.

\subsubsection{Gateway revocation}

The authorization epoch~$\varepsilon$ embedded in each packet
header provides a revocation signal.  When the controller bumps
the epoch (via \texttt{bpf\_map\_update\_elem} on the gateway's
epoch array), the TC ingress program detects the mismatch between
the packet's epoch and the current global epoch and drops the
stale packet.
Flow-cache entries carry the epoch at insertion time; a stale cache
entry is treated as a miss, forcing a fresh LPM lookup against the
updated policy maps.
This provides sub-millisecond revocation: the median epoch-bump
latency measured in Section~\ref{sec:eval:wireguard} is
190\,\textmu{}s, with new connections blocked within
146\,\textmu{}s.

\paragraph{Tag-carrier limitations.}
The IPv4 identification field is consumed by IP fragmentation;
the TC egress program sets DF and the prototype assumes a
tunnel MTU that avoids fragmentation (WireGuard enforces this
by default).
Overwriting TOS/traffic class clobbers DSCP and ECN bits;
the gateway clears tag fields on allow, restoring zeros, which
is acceptable for tunnel-interior hops but discards any
upstream DSCP markings.
The IPv6 flow label is 20~bits, limiting \texttt{app\_index}
to $2^{20}$ principals, sufficient for endpoint policy but
may conflict with ECMP hashing on intermediate routers; the
WireGuard outer header carries its own flow label, so inner
reuse does not affect outer-path ECMP.
A production deployment in a DSCP-sensitive network could avoid
header-field overloading entirely by using a dedicated
encapsulation shim (e.g., a Geneve TLV or a fixed-offset
preamble inside the WireGuard tunnel) at the cost of additional
per-packet overhead; the prototype uses header-field tagging
because it requires no tunnel-format changes and is sufficient
for the evaluated scenarios.

\section{Security Analysis}
\label{sec:security}

We now analyze the security properties ProcRoute provides under
the threat model: enforcement
completeness, the least-privilege guarantee of default-deny,
the integrity of cgroup-based principal binding, and the
boundaries of these guarantees.

\subsection{Enforcement completeness at socket operations}

ProcRoute interposes on outbound connections via four eBPF
cgroup socket-address hooks (\texttt{cgroup/connect4/6},
\texttt{cgroup/sendmsg4/6}), invoked \emph{before} the kernel
processes the socket address~\cite{kernel_bpf_cgroup}.
Returning \texttt{BPF\_DROP} fails the call with \texttt{EPERM}
before any packet is emitted.
All user-space TCP connections pass through \texttt{connect()},
and all UDP transmissions through \texttt{sendmsg()}-family calls;
cgroup BPF programs are inherited by descendant
cgroups~\cite{kernel_cgroup2}, so the entire ProcRoute subtree
is covered without re-attachment.
Raw sockets (\texttt{SOCK\_RAW}) are not mediated but require
\texttt{CAP\_NET\_\allowbreak RAW}, excluded by the threat model.

\subsection{Default-deny and least privilege}

The decision function yields
default-deny on internal prefixes: an unbound process
($\mathit{prin} = \bot$) has zero internal access regardless of
UID or routing table, and a bound process can reach only its
granted resources $R_p$.
External destinations are never blocked, preserving split-tunnel
Internet access.

\subsection{Cgroup binding integrity}
\label{sec:cgroup-integrity}

ProcRoute's security rests on two properties:
(1)~an unprivileged adversary cannot join an authorized cgroup,
and (2)~the controller binds only intended processes.

\paragraph{Atomic placement (primary path).}
On Linux~$\geq$5.7, the prototype uses \texttt{clone3()} with the
\texttt{CLONE\_INTO\_CGROUP} flag as the \emph{primary, production
path} for process placement.
The controller passes the target cgroup's file descriptor to
\texttt{clone3()}, and the kernel places the child into the
designated cgroup atomically at creation time, before any
user-space code in the child executes.
This eliminates the window between process creation and cgroup
migration entirely; no interposition gap exists in which the child
could issue a \texttt{connect()} under the wrong principal.

Forked children inherit cgroup membership; \texttt{exec\_hash}
verification invalidates
grants on \texttt{exec()}, but \texttt{fork()}-only children
retain grants unconditionally.

\paragraph{Binding integrity statement.}
An unprivileged process cannot migrate itself into a ProcRoute
application cgroup for three reasons.
First, ProcRoute's controller creates all application cgroups
under \texttt{/sys/fs/cgroup/procroute/} with root ownership
(\texttt{uid=0, gid=0}); writing a PID to \texttt{cgroup.procs}
requires write permission on that file, which unprivileged
processes do not possess~\cite{kernel_cgroup2}.
Second, the controller does \emph{not} delegate any ProcRoute
cgroup subtree~\cite{kernel_cgroup2_delegation}: no
\texttt{cgroup.delegate} file is set, and no non-root user
receives write access to any node in the ProcRoute hierarchy.
Third, even if an adversary creates a cgroup elsewhere in the
hierarchy (e.g., under a user-delegated subtree), that cgroup
has no entry in the \texttt{cgroup\_to\_app} BPF map; the BPF
hook resolves $\mathit{prin} = \bot$ and returns \textsc{deny}.
These three properties together ensure that cgroup-based principal
binding cannot be forged by an unprivileged adversary.

\subsection{Tag integrity in split architecture}
\label{sec:tag-integrity}

When ProcRoute operates in split mode, the security argument extends
from cgroup binding integrity to tag integrity across the tunnel.

\paragraph{\texttt{SO\_MARK} protection.}
The client BPF hooks set \texttt{SO\_MARK} within the kernel's BPF
execution context; the mark is subsequently copied into packet
headers by the TC egress program.
Setting \texttt{SO\_MARK} from user space requires
\texttt{CAP\_NET\_ADMIN}, which the threat model excludes.
An unprivileged adversary therefore cannot forge a mark: its
sockets carry $\texttt{SO\_MARK} = 0$, producing untagged packets
that the gateway drops unconditionally.

\paragraph{Gateway enforcement completeness.}
The TC ingress program on \texttt{wg0} executes on every packet
delivered by the WireGuard interface.  Because \texttt{wg0} is the
only ingress path for tunnel traffic, and the TC hook fires before
the kernel routes the packet, there is no bypass path.
Combined with WireGuard's authenticated encryption, which prevents
off-path injection or modification of inner-packet headers, the
gateway sees exactly the tags stamped by the client.
WireGuard also binds each peer's public key to its allowed tunnel
IP, so the inner source address~$s$
used in the per-peer grant lookup is cryptographically
authenticated; one peer cannot spoof another's tunnel IP.

\paragraph{Composition argument.}
The combination of (i)~mandatory client tagging (unbound processes
produce tag~$= 0$), (ii)~unforgeability of \texttt{SO\_MARK} by
unprivileged processes, (iii)~WireGuard authenticated encryption
binding tag to source peer, and (iv)~gateway enforcement of a
per-peer policy $\mathcal{R}_{\mathit{gw}}(s_k, a)$ keyed on
tunnel source IP and app\_index implies that the reachable set of
any process under the split architecture is bounded by the
gateway's grant set for that peer.
When the gateway grant set is a subset of the client's local
grants (the expected configuration), this gives
$\mathit{Reach}_{\mathit{split}}(s_k, p) \subseteq
 \mathit{Reach}_{\mathit{local}}(p)$ for all peers~$s_k$ and
principals~$p$ (including $\bot$), as stated in
Section~\ref{sec:model:split}.

\subsection{Residual attack surface}
\label{sec:residual}

\paragraph{In-process compromise.}
Code execution \emph{within} an authorized process (memory
corruption, malicious extensions, library injection) inherits the
victim's cgroup membership and grants.  This is inherent to any
process-granularity boundary and requires complementary controls
(sandboxing, exploit mitigation).
Complementary mechanisms operating within the process
boundary, such as \texttt{seccomp}-BPF profiles restricting
\texttt{connect()} argument ranges, Landlock network rules
(Linux~$\geq$6.7), and CFI-instrumented binaries preventing
control-flow hijacking, can significantly raise the bar for
injected code reaching a socket call.
For browser principals specifically, the existing renderer
process sandbox already provides an independent enforcement
boundary that does not rely on cgroup membership.

\paragraph{Privileged attacker.}
A \texttt{root}/\texttt{CAP\_SYS\_ADMIN} attacker can detach BPF
programs or modify maps; the threat model excludes this case.

\paragraph{DNS-based bypasses.}
ProcRoute enforces on IP prefixes, not DNS names.  Resolved IPs
outside $\mathcal{I}$ reach no internal resource; those inside
$\mathcal{I}$ are subject to the decision function.  DNS integrity
is an orthogonal concern.

\subsection{Policy shrink and revocation behavior}
\label{sec:revocation}

ProcRoute intervenes at \texttt{connect()} and
\texttt{sendmsg()} time; it does not monitor sockets after
establishment.
Consequently, when a policy update \emph{shrinks} a principal's
grant set (e.g., removing a prefix from~$R_p$), TCP connections
that were authorized under the prior policy and are already in
the \texttt{ESTABLISHED} state persist until they are closed by
the application or time out.
New connections to the revoked prefix are denied immediately by
the updated BPF maps, but the in-flight socket retains its
kernel state and could continue to exchange data.

The prototype implements G1 (connect-time) and two mechanisms
toward G2 (data-path revocation):
(1)~an \emph{authorization epoch} stored in a BPF array map: on
policy update, the controller increments the epoch; each
\texttt{connect} hook stamps the socket cookie with the current
epoch, and each \texttt{sendmsg} hook denies sends on stale
sockets (immediate UDP revocation);
(2)~a \emph{background sweeper} that scans \texttt{/proc/net/tcp}
at 1\,s intervals and issues \texttt{ss~--kill}
(\texttt{SOCK\_DESTROY}) for stale TCP sockets.
A production deployment could add a \texttt{sock\_ops} hook for
per-segment enforcement; We believe the sweeper program is
sufficient for the evaluation.

\section{Evaluation}
\label{sec:eval}
ProcRoute's enforcement spans two layers: cgroup socket-address
hooks that resolve principals and evaluate grants at
\texttt{connect()}\slash\texttt{sendmsg()} time, and, in split
mode, TC programs that tag and enforce packets across a
WireGuard tunnel.
We evaluate these layers separately and in composition.
Single-machine loopback microbenchmarks (Q1--Q6) isolate the
cgroup hook overhead that is common to both local-only and
client-side tagging modes, independent of network and encryption
costs.
A two-machine WireGuard deployment then measures the end-to-end
cost of the full stack (client tagging, tunnel encryption, and
gateway TC enforcement) including throughput, latency,
multi-stream scaling, flow-cache effectiveness, gateway policy
scaling, and revocation.

\subsection{Experimental setup}
\label{sec:eval:setup}

\paragraph{Single-machine microbenchmark (loopback).}
All Q1--Q6 benchmarks ran on a dedicated bare-metal server:
AMD EPYC 7443P (24~cores, 2.85\,GHz base), 64\,GB DDR4,
Ubuntu~24.04.3~LTS, cgroup~v2 unified hierarchy.
Clients and listeners were pinned to separate cores.
Internal-prefix traffic used a loopback alias, causing the BPF
hook to walk the full pipeline, thus isolating the cgroup BPF hook overhead from network and
encryption costs.

\paragraph{WireGuard deployment (two machines).}
Each machine is an AMD Ryzen~7 7700X
(8~cores/16~threads, 4.5\,GHz, 61\,GB RAM) running
Ubuntu~24.04~LTS, connected via
a WireGuard tunnel over public IPv4.
On the \emph{client}, cgroup \texttt{connect4/6} hooks tag
authorized flows with \texttt{SO\_MARK} and IPv6 flow
labels.  On the \emph{gateway}, a TC ingress BPF program on
\texttt{wg0} enforces per-principal policy with an optional
flow cache.
Five configurations are tested:
\emph{wg\_baseline} (no BPF),
\emph{wg\_nftables} (nftables cgroup match),
\emph{wg\_tag\_only} (client hooks only),
\emph{wg\_enforce\_nocache} (full enforcement,
no cache), and \emph{wg\_enforce\_cache} (full
enforcement with flow cache).

\subsection{Single-machine microbenchmark summary}
\label{sec:eval:micro-summary}

To isolate BPF hook overhead from network and encryption
costs, we ran six microbenchmark experiments (Q1--Q6) on a
loopback path using a dedicated bare-metal server (AMD EPYC
7443P, 24~cores,
64\,GB RAM).
Table~\ref{tab:micro-summary} summarizes the key results; full
methodology, per-experiment tables, and analysis are in
Appendix~\ref{sec:appendix:micro}.

\begin{table}[t]
\centering
\caption{Single-machine microbenchmark summary (loopback).
  Full details in Appendix~\ref{sec:appendix:micro}.}
\label{tab:micro-summary}
\begin{tabular}{p{3.2cm} p{4.4cm}}
\toprule
Experiment & Key result \\
\midrule
Q1: Hook overhead
  & p50 internal-allow = 25.6\,\textmu{}s (+2.7\,\textmu{}s);
    throughput = 31.2\,Gbps \\
Q2: Pivot prevention
  & 82/82 unauthorized attempts blocked \\
Q3: Policy scaling
  & Flat p50 at 28.7\,\textmu{}s with 4\,096 prefixes \\
Q4: Update safety
  & 0 transient allows / 1.24M attempts \\
Q5: nftables comparison
  & ProcRoute 25.6\,\textmu{}s vs.\ nftables 30.3\,\textmu{}s \\
Q6: Revocation
  & UDP epoch deny $\sim$27\,\textmu{}s; TCP sweeper $\sim$11\,ms \\
\bottomrule
\end{tabular}
\end{table}

The loopback results confirm that BPF hook overhead is negligible
(2.7\,\textmu{}s on the internal-allow path), flat across policy
sizes from 4 to 4\,096 prefixes and 5 to 200 principals, and
structurally safe during policy reloads.
ProcRoute blocked all 82 pivot attempts while retaining full
external connectivity, recorded zero transient allows across
1.24~million connection attempts during policy reload, and
revocation is immediate for UDP and sub-15\,ms for TCP.

These cgroup hook costs carry directly into the split
architecture: the client-side \texttt{connect4/6} hooks execute
the same principal-resolution and \texttt{SO\_MARK} stamping
logic, so the 2.7\,\textmu{}s overhead measured here is the
per-connection cost on the client in split mode as well.
The WireGuard evaluation below adds the TC tagging and gateway
enforcement layers on top of this baseline.

\subsection{End-to-end WireGuard deployment}
\label{sec:eval:wireguard}

This section evaluates ProcRoute over a realistic two-machine
WireGuard tunnel similar to an enterprise deployment.

\subsubsection{Throughput and connect latency}
Table~\ref{tab:wg:throughput} summarizes single-stream TCP throughput
(10\,s \texttt{iperf3}, 10 trials) and TCP \texttt{connect()} latency
over the WireGuard tunnel.
All ProcRoute configurations match the WireGuard baseline at
8\,878\,Mbps.
The nftables cgroup-matching configuration achieves only
7\,836\,Mbps, 12\% lower than baseline, due to per-packet netfilter rule evaluation
traversal on the gateway.
Connect latency at p50 is 89.8\,\textmu{}s for the baseline and
93.4--93.7\,\textmu{}s for full enforcement
(+3.6\,\textmu{}s), well within the WireGuard round-trip noise
floor.  The nftables configuration adds 7.5\,\textmu{}s.

\begin{table}[t]
\centering
\caption{WireGuard single-stream throughput and connect latency
  (10 trials each).}
\label{tab:wg:throughput}
\small
\begin{tabular}{lrrrr}
\toprule
Config.\ & Tput & p50 & p95 & p99 \\
        & (Mbps) & (\textmu{}s) & (\textmu{}s) & (\textmu{}s) \\
\midrule
wg\_baseline          & 8\,878 & 89.8 & 98.7 & 100.9 \\
wg\_nftables          & 7\,836 & 97.3 & 100.6 & 101.9 \\
wg\_tag\_only         & 8\,878 & 90.3 & 98.4 & 100.5 \\
wg\_enforce\_nocache  & 8\,878 & 93.4 & 99.0 & 100.7 \\
wg\_enforce\_cache    & 8\,876 & 93.7 & 101.2 & 102.3 \\
\bottomrule
\end{tabular}
\end{table}

\subsubsection{Multi-stream scaling}

Table~\ref{tab:wg:multistream} reports aggregate throughput with
8 and 16 parallel \texttt{iperf3} streams.
At 16 streams ProcRoute enforcement without cache achieves
8\,298\,Mbps, within 2.3\% of the baseline;
with the flow cache the gap is 2.8\%.
The nftables configuration drops to 7\,661\,Mbps, 9.8\% below
baseline, confirming that per-packet netfilter rule evaluation overhead compounds
unfavorably with stream count, which is critical in enterprise environments.

\begin{table}[t]
\centering
\caption{WireGuard multi-stream throughput (Mbps, 10 trials).}
\label{tab:wg:multistream}
\begin{tabular}{lrr}
\toprule
Configuration & 8 streams & 16 streams \\
\midrule
wg\_baseline          & 8\,786 & 8\,497 \\
wg\_nftables          & 7\,780 & 7\,661 \\
wg\_tag\_only         & 8\,780 & 8\,439 \\
wg\_enforce\_nocache  & 8\,757 & 8\,298 \\
wg\_enforce\_cache    & 8\,764 & 8\,260 \\
\bottomrule
\end{tabular}
\end{table}

\subsubsection{Flow-cache effectiveness}

The gateway TC program optionally caches per-flow allow
decisions keyed on (src\_ip, dst\_ip, dst\_port, protocol)
so that subsequent packets skip the full LPM + principal
lookup.
Hit rates vary by workload: many short flows (1\,000
four-packet flows) yield 75\% cache hits, while few long
flows (4 persistent streams) reach $>$99.99\% (650\,908
hits vs.\ 4 misses).
For long flows the cache amortizes nearly all gateway
lookups; for short flows the first packet incurs a full
lookup but subsequent packets hit the cache.

\subsubsection{Gateway policy scaling}

We varied the number of principals (10--500) with 10 prefixes each
and measured gateway startup time and sustained throughput
(Table~\ref{tab:wg:scaling}).
Startup time is flat at $\sim$204\,ms across all sizes; throughput
remains at 8\,855--8\,858\,Mbps regardless of whether the policy
contains 100 or 5\,000 total prefixes.
The LPM trie lookup is $O(\text{prefix length})$ and independent of
map population, consistent with the flat scaling observed in the
single-machine results (Table~\ref{tab:scaling}).

\begin{table}[t]
\centering
\caption{Gateway policy scaling (3 trials).}
\label{tab:wg:scaling}
\begin{tabular}{rrrr}
\toprule
Principals & Total prefixes & Startup (ms) & Tput (Mbps) \\
\midrule
10  & 100   & 204.4 & 8\,858 \\
50  & 500   & 204.5 & 8\,856 \\
100 & 1\,000 & 204.6 & 8\,858 \\
250 & 2\,500 & 204.3 & 8\,857 \\
500 & 5\,000 & 204.8 & 8\,855 \\
\bottomrule
\end{tabular}
\end{table}

\subsubsection{WireGuard revocation}

Table~\ref{tab:wg:revocation} reports revocation latency over the
WireGuard tunnel.  Epoch bumps complete in 190\,\textmu{}s (median),
new connections are blocked within 146\,\textmu{}s, and
steady-state connect latency is 102\,\textmu{}s, all
sub-millisecond.  Flow-cache entries are invalidated atomically on
epoch bump, confirming
that cached allow decisions do not survive revocation.

\begin{table}[t]
\centering
\caption{WireGuard revocation latency (10 trials).}
\label{tab:wg:revocation}
\begin{tabular}{lrr}
\toprule
Metric & Median (\textmu{}s) & Range (\textmu{}s) \\
\midrule
Epoch bump            & 190 & 189--195 \\
New connection block  & 146 & 137--153 \\
Steady-state connect  & 102 & 94--109 \\
\bottomrule
\end{tabular}
\end{table}

\subsubsection{Summary}

The two-machine WireGuard results confirm that ProcRoute's
architecture scales from single-machine loopback to a realistic
encrypted tunnel with negligible overhead.  Client-side cgroup hooks
add 3.6\,\textmu{}s to connect latency; gateway-side TC enforcement
has no measurable throughput impact.  In contrast, the nftables cgroup-matching
configuration incurs a 12\% throughput penalty, validating ProcRoute's
design choice of flow-level tagging over per-packet rule evaluation.
Revocation remains sub-millisecond, and policy scaling is flat to
5\,000 prefixes, consistent with the single-machine results.

\section{Discussion}
\label{sec:discussion}

\subsection{Deployment model}
ProcRoute can be deployed incrementally: (1)~observe-only mode to
log which processes attempt internal access; (2)~enforce for
high-risk prefixes first; (3)~expand per-application allowlists as
needed.
For HTTP/S applications already mediated by a ZT proxy, ProcRoute
restricts direct IP access to internal subnets while allowing
proxy endpoints.
For non-HTTP protocols,
ProcRoute provides least-privilege route gating without protocol
translation.

The WireGuard evaluation
demonstrates that ProcRoute's split client/gateway architecture
operates at line rate over an encrypted tunnel, with negligible
connect-latency overhead and flat policy scaling to 5\,000
prefixes. This confirms that the deployment model is practical for
production quality deployments in enterprise environments.
Although the prototype and evaluation use WireGuard as the
tunnel substrate, the split architecture is not WireGuard-specific:
it requires only (i)~a kernel-visible tunnel interface on which
a TC program can inspect and tag packets, and (ii)~authenticated
encryption that binds inner headers to a peer identity.
Any VPN or overlay that exposes a \texttt{tun}/\texttt{tap}
interface (e.g., OpenVPN, IPsec \texttt{xfrm} interfaces,
Tailscale, Cloudflare WARP) satisfies both conditions, and the
same TC ingress/egress programs could be attached with only
interface-name and peer-identity mapping changes.

\subsection{Future work}
We see three natural extensions: automatic policy synthesis from observed
traffic, DNS-name-aware selectors to reduce over-broad prefix
grants, and cross-platform prototypes on Windows~(WFP callout
drivers) and macOS~(Network Extension content filters).
ProcRoute's access-control model is OS-agnostic; the decision
function requires only a principal-binding mechanism, a synchronous
enforcement hook, and an audit channel, all available on Windows and macOS.

\section{Related Work}

\paragraph{nftables cgroup matching.}
nftables \texttt{socket~cgroupv2} matching~\cite{nftables_wiki_cgroups}
provides basic per-cgroup destination filtering.
ProcRoute contributes a formal model, unified policy surface, monotonic
composition (Theorem~\ref{thm:monotonicity}), and structured
audit.

\paragraph{systemd \texttt{IPAddressAllow/Deny}.}
systemd (v235+) attaches eBPF programs to a unit's cgroup for
IP-prefix filtering~\cite{systemd_resource_control}, limited
to long-running services with static prefixes, and unable to
cover interactive desktop applications, dynamic prefix changes,
or cross-unit policy.

\paragraph{Process confinement.}
SELinux~\cite{selinux},
AppArmor~\cite{apparmor},
Capsicum~\cite{capsicum},
\texttt{seccomp}~\cite{seccomp}, and
Landlock~\cite{landlock} confine processes via labels,
descriptors, or syscall classes but do not express
destination-prefix policies for dynamic split-tunnel
route sets.

\paragraph{Zero trust and eBPF ecosystem.}
ZT architectures~\cite{nist_sp80020r1,beyondcorp} authenticate
users/devices; ProcRoute adds process-level enforcement.
eBPF projects (Cilium~\cite{cilium}, Tetragon~\cite{tetragon})
target containers, not endpoint route authorization.

\section{Conclusion}
Enterprises rely on split tunneling for good reason, but it turns
every device into a potential pivot.  Existing OS mechanisms like
nftables cgroup matching and systemd \texttt{IPAddressAllow} can
approximate per-process filtering, yet they remain low-level
primitives with no formal policy semantics, no composition
guarantees, and no structured audit trail.
ProcRoute fills this gap: by modeling routes as protected resources
and processes as principals governed by a total decision function,
organizations can preserve split-tunnel usability while confining
internal-route access to authorized processes with auditable,
formally grounded guarantees.
Our evaluations
confirm that ProcRoute matches baseline throughput, adds negligible
connect latency, scales flat to thousands of prefixes, and revokes
access in sub-millisecond time, all without per-packet rule
evaluation overhead.

\bibliographystyle{ACM-Reference-Format}
\bibliography{refs}

\appendix

\section{Single-Machine Microbenchmark Details}
\label{sec:appendix:micro}

This appendix provides the full methodology and results for the
six single-machine microbenchmarks summarized in
Section~\ref{sec:eval:micro-summary}.
All experiments ran on the loopback setup described in
Section~\ref{sec:eval:setup}.

\subsection{Q1: Connection-time enforcement overhead}
\label{sec:appendix:q1}

We measured TCP \texttt{connect()} latency under three
conditions (Table~\ref{tab:overhead:e2e}):
(1)~\emph{baseline}, outside ProcRoute (no hook);
(2)~\emph{external-miss}, single LPM miss
(\texttt{127.0.0.1});
(3)~\emph{internal-allow}, full
pipeline (LPM hit $\to$ cgroup-to-principal $\to$
per-app LPM $\to$ port check).
The BPF-only deny overhead (5.7\,\textmu{}s, no TCP
handshake) is reported separately in the nftables
comparison (Section~\ref{sec:appendix:q5}).

\paragraph{Results.}
Table~\ref{tab:overhead:e2e} reports end-to-end
\texttt{connect()} latency and isolates BPF overhead
($\Delta$p50 = condition p50 $-$ baseline p50).
The external-miss path adds 1.8\,\textmu{}s (one LPM
lookup that misses).
The internal-allow path adds 2.7\,\textmu{}s (three map
lookups plus port/protocol check).
Tail latencies (p90, p99) are tight on bare metal: the baseline
p99 is 32\,\textmu{}s.

TCP throughput was 31.2\,Gbps (internal-allow) vs.\ 34.2\,Gbps
baseline (9\% overhead); the hook fires once at \texttt{connect()}
time and imposes no per-packet cost.

\begin{table}
\centering
\caption{TCP \texttt{connect()} latency (\textmu{}s).
  $\Delta$p50 = condition $-$ baseline.}
\label{tab:overhead:bpf}\label{tab:overhead:e2e}\label{tab:overhead}
\begin{tabular}{lrrrr}
\toprule
Condition & p50 & $\Delta$p50 & p90 & p99 \\
\midrule
Baseline       &  22.9 &   ---  &  24.8 &  32.6 \\
External-miss  &  24.7 &  1.8   &  27.1 &  35.1 \\
Internal-allow &  25.6 &  2.7   &  27.2 &  34.9 \\
\bottomrule
\end{tabular}
\end{table}

\subsection{Q2: Pivot prevention}
\label{sec:appendix:q2}

All 82 connection attempts from an unauthorized process (no
principal binding) across six service categories were denied;
Table~\ref{tab:pivot} breaks down the test matrix.

\begin{table}[t]
\centering
\caption{Pivot-prevention test matrix (unauthorized process, all denied).}
\label{tab:pivot}
\small
\begin{tabular}{lp{2.6cm}rl}
\toprule
Service & Target subnet & \# & Proto \\
\midrule
SSH (T1021.004)   & 10.0.0.0/24, 10.250.0.0/24 & 30 & TCP \\
HTTPS (T1071.001) & 10.0.0.0/24             & 20 & TCP \\
RDP (T1021.001)   & 10.0.0.0/24             & 10 & TCP \\
PostgreSQL (T1210)& 10.0.0.0/24, 10.0.1.0/24 & 10 & TCP \\
Alt-HTTP          & 10.0.0.0/24             & 10 & TCP \\
DNS (internal)    & 10.0.0.0/24, 10.0.1.0/24 &  2 & UDP \\
\midrule
\textbf{Total}    &                        & \textbf{82} & \\
\bottomrule
\end{tabular}
\end{table}

\subsection{Q3: Policy-update latency and scaling with policy size}
\label{sec:appendix:q3}

ProcRoute's lifecycle layer
(Section~\ref{sec:lifecycle}) claims that policy changes
take effect by atomically repopulating BPF maps via
\texttt{bpf\_map\_\allowbreak update\_elem}.
We evaluate two aspects: (A)~whether map-population
time remains practical as policy size grows, and
(B)~whether lookup latency in the data-plane hook
degrades with larger maps.

\paragraph{Methodology.}
Synthetic policies with $N \in \{4\ldots4096\}$ prefixes and
$M \in \{5\ldots200\}$ principals (48~configurations).
(A)~Update latency: wall-clock time from daemon launch to all maps
populated (3~trials).
(B)~Lookup latency: 2\,000 TCP \texttt{connect()} iterations
(200~warmup) per configuration (3~trials).
Table~\ref{tab:scaling} reports medians for $M = 5$.

\begin{table}[t]
\centering
\caption{Policy-update latency and internal-allow lookup p50
  vs.\ number of internal prefixes~($N$).}
\label{tab:scaling}
\begin{tabular}{rrr}
\toprule
$N$ (prefixes) & Update (ms) & Lookup p50 (\textmu{}s) \\
\midrule
     4   & 107.0  &  25.5  \\
    16   & 107.1  &  26.1  \\
    64   & 107.1  &  26.4  \\
   256   & 107.3  &  26.2  \\
   512   & 107.5  &  27.3  \\
 1\,024  & 107.4  &  26.7  \\
 2\,048  & 107.2  &  28.4  \\
 4\,096  & 107.2  &  28.7  \\
\bottomrule
\end{tabular}
\end{table}

\paragraph{Results.}
Update latency is constant at $\sim$107\,ms across all
configurations from 4 to 4\,096 prefixes, dominated by
fixed costs (BPF loading, map creation, cgroup provisioning);
the per-entry insertion cost is negligible.
Internal-allow lookup p50 is flat at 25--29\,\textmu{}s from
$N = 4$ to $N = 4096$, confirming that LPM trie lookup time
depends on prefix length (at most 32~bits for IPv4), not on
the number of entries.
Varying the number of principals from 5 to 200 had no
measurable effect on lookup latency (all clean medians within
25--29\,\textmu{}s).
These latencies are well within the budget of a VPN route
renegotiation.

\subsection{Q4: Update safety under load}
\label{sec:appendix:q4}

A tight \texttt{connect()} loop from an unauthorized process ran
through 10 policy-reload cycles, comparing ProcRoute (BPF stays
attached; map entries updated in place) with nftables
(\texttt{nft flush ruleset \&\& nft -f}).

\paragraph{Results.}
ProcRoute recorded zero transient allows across 1\,240\,367
attempts (Table~\ref{tab:update-safety}).
nftables' \texttt{flush ruleset} reverts to default-accept:
3\,085 of 3\,485 attempts (88.5\%) succeeded during the window,
confirming the structural difference stated in
Theorem~\ref{thm:update-safety}.

\begin{table}[t]
\centering
\caption{Transient allows during 10 policy-reload cycles.}
\label{tab:update-safety}
\begin{tabular}{lrr}
\toprule
Mechanism & Total attempts & Transient allows \\
\midrule
ProcRoute (BPF attached)  & 1\,240\,367 & 0 \\
nftables (flush + reload) &     3\,485  & 3\,085 \\
\bottomrule
\end{tabular}
\end{table}

\subsection{Q5: Structural comparison with nftables}
\label{sec:appendix:q5}

We implemented an equivalent nftables configuration (nine
\texttt{socket\allowbreak\,cgroupv2} rules mirroring the
five-principal policy) and ran the same latency
microbenchmark; nftables also denied all 82~pivot attempts
(Table~\ref{tab:pivot}) with comparable median latency
(Table~\ref{tab:nft-comparison}).

\begin{table}[t]
\centering
\caption{p50 TCP \texttt{connect()} latency (\textmu{}s): nftables
  vs.\ ProcRoute.  $\dagger$\,nftables \texttt{drop} causes a
  TCP timeout, not \texttt{EPERM}.}
\label{tab:nft-comparison}
\begin{tabular}{lrr}
\toprule
Condition & nftables & ProcRoute \\
\midrule
Baseline       & 28.6 &  22.9 \\
External-miss  & 29.7 &  24.7 \\
Internal-allow & 30.3 &  25.6 \\
Internal-deny  & $\dagger$ &  5.7 \\
\bottomrule
\end{tabular}
\end{table}

\paragraph{Structural differences.}
(1)~nftables evaluates rules in a linear chain where
ordering determines the verdict; ProcRoute's conjunctive
composition eliminates rule-ordering failures by
construction.
(2)~ProcRoute's deny events include PID, command name,
and principal; nftables logs only packet metadata.
(3)~ProcRoute optionally gates on per-binary SHA-256
verification; nftables has no binary-identity mechanism.

\subsection{Q6: Revocation latency}
\label{sec:appendix:q6}

Table~\ref{tab:revocation} summarizes the two revocation
mechanisms.  After an epoch increment, all 1\,000 subsequent
UDP \texttt{sendto()} calls were denied within
$\sim$27\,\textmu{}s.  The TCP sweeper terminated 50~established
connections in $\sim$11\,ms.  Epoch-check overhead is negligible
(+1.6\,\textmu{}s over baseline).

\begin{table}[t]
\centering
\caption{Revocation latency and epoch-check overhead.}
\label{tab:revocation}
\small
\begin{tabular}{llr}
\toprule
Mechanism & Metric & Value \\
\midrule
UDP epoch     & First deny     & $\sim$27\,\textmu{}s \\
UDP epoch     & Denied/total   & 1\,000/1\,000 \\
TCP sweeper   & 50-sock revoke & $\sim$11\,ms \\
Epoch ovhd.   & p50/p90/p99    & 24.5/27.5/42.9\,\textmu{}s \\
\bottomrule
\end{tabular}
\end{table}

\paragraph{Limitations.}
The current epoch mechanism has two limitations:
(1)~the UDP epoch check applies to all sockets stamped at an
older epoch, not only those whose specific grant was revoked; a
selective per-principal epoch would reduce false revocations in
multi-principal updates;
(2)~TCP revocation relies on a userspace sweep rather than an
in-kernel \texttt{sock\_ops} hook, bounding the worst-case
revocation delay to the sweep interval (1\,s in the prototype).
A production system could combine both: the epoch mechanism for
immediate UDP/first-packet enforcement and a \texttt{sock\_ops}
program for per-segment TCP enforcement.
The destination-disjoint assumption (Section~\ref{sec:lpm-overlaps})
prevents a single principal from holding overlapping prefixes with
different port constraints, e.g., allowing \texttt{10.0.0.0/8}
on port~443 while restricting \texttt{10.1.0.0/16} to port~22.
Enterprise firewall audits suggest that 15--25\,\% of zone-based
rulesets contain such overlapping-prefix, port-differentiated
entries~\cite{wool2004quantitative,wool2010trends}.
This restriction could be lifted by replacing the single LPM
lookup with a \emph{cascaded trie}: for each principal, a
longest-match hit is followed by a secondary hash-map lookup
keyed on $\langle\mathit{prefix},\mathit{proto}\rangle$, adding
one extra map access ({\raise.17ex\hbox{$\scriptstyle\sim$}}40\,ns)
per connection but supporting arbitrary per-prefix port
predicates.

\end{document}